\documentclass[prl,twocolumn,showpacs,superscriptaddress,amsmath,amssymb]{revtex4-2}
\usepackage{graphicx}
\usepackage{dcolumn}
\usepackage{bm}
\usepackage{hyperref}
\usepackage{braket}
\usepackage{xcolor}
\definecolor{darkgreen}{rgb}{0.0, 0.5, 0.13}
\usepackage{hyperref}
\usepackage{cleveref}
\usepackage{lipsum}
\usepackage{enumitem}

\usepackage{graphicx}
\usepackage{tikz}

\usepackage{changes} 
\definechangesauthor[name=Thomas,color=magenta]{TS}

\DeclareMathOperator{\Tr}{Tr}
\DeclareMathOperator*{\argmax}{arg\,max}

\usepackage{amsthm}

\newtheorem{theorem}{Theorem}
\newtheorem{lemma}{Lemma}
\newtheorem{corollary}{Corollary}
\crefname{lemma}{Lemma}{Lemmas}
\newtheorem{proof-sketch}{Proof Sketch}
\usepackage{algorithmicx}
\usepackage[Algorithm,ruled]{algorithm} 
\usepackage{algorithm}
\usepackage{algcompatible}

\begin{document}

\title{Quantifying the intrinsic randomness of quantum measurements}
\author{Gabriel Senno}
\affiliation{ICFO-Institut de Ciencies Fotoniques, The Barcelona
Institute of Science and Technology, 08860 Castelldefels,
Barcelona, Spain}
\affiliation{Quside Technologies S.L., C/Esteve Terradas 1, 08860 Castelldefels,
Barcelona, Spain}
\author{Thomas Strohm}
\affiliation{Corporate Research, Robert Bosch GmbH, 71272 Renningen, Germany}
\author{Antonio Ac\'in}
\affiliation{ICFO-Institut de Ciencies Fotoniques, The Barcelona
Institute of Science and Technology, 08860 Castelldefels,
Barcelona, Spain}
\affiliation{ICREA-Institucio Catalana de Recerca i Estudis Avan\c cats, Lluis Companys 23, 08010 Barcelona, Spain}

\begin{abstract}
Intrinsic quantum randomness is produced when a projective measurement on a given basis is implemented on a pure state that is not an element of the basis. The prepared state and implemented measurement are perfectly known, yet the measured result cannot be deterministically predicted. In realistic situations, however, measurements and state preparation are always noisy, which introduces a component of stochasticity in the outputs that is not a consequence of the intrinsic randomness of quantum theory. Operationally, this stochasticity is modelled through classical or quantum correlations with an eavesdropper, Eve, whose goal is to make the best guess about the outcomes produced in the experiment. In this work, we study Eve's maximum guessing probability when she is allowed to have correlations with, both, the state and the measurement. We show that, unlike the case of projective measurements (as it was already known) or pure states (as we prove), in the setting of generalized measurements and mixed states, Eve's guessing probability differs depending on whether she can prepare classically or quantumly correlated strategies. %

\end{abstract}

\maketitle

\section{Introduction}
Quantum theory contains a form of randomness that is not the result of ignorance or any stochastic behaviour. For instance, according to the theory, the result of implementing a spin measurement along the $x$ direction on a spin state pointing in the $+z$ direction is fully unpredictable. This is despite the fact that the description of the experiment within the theory is complete, in the sense that the prepared state and implemented measurement are perfectly known, without any stochastic component. 
This form of randomness is \emph{intrinsic} to quantum theory and impossible in classical physics~\cite{randreview1,randreview2}. Beyond fundamental considerations, it is also the key element behind any quantum random-number generator~(QRNG).

In real life implementations, however, measurements are never projective and states are never pure. Noise and imperfections introduce an unavoidable element of \emph{stochasticity} that produces an apparent randomness that is not intrinsic to quantum theory. Therefore, it is a fundamental problem to design the tools to estimate the correct amount of intrinsic quantum randomness produced in a quantum experiment. This question is of relevance from a quantum foundations viewpoint, but also for the proper design of QRNGs. In fact, the natural and operational way to model the stochasticity in the components of the setup is through classical or quantum correlations with an external observer, Eve, who can also be interpreted as an eavesdropper and whose goal is to make the best \emph{guess} about the outcomes produced in the experiment. The correlations with Eve are often named (classical or quantum) side information.


So far, the scenario that has mostly been considered in the literature is the one in which all the stochasticity comes from the prepared quantum state. That is, the state of the system is no longer pure, but the measurement is still assumed to be projective. The main goal of this work is to study \emph{Eve's guessing probability} about the outcomes of a quantum measurement when she is allowed to have correlations with, both, the state and the measurement. We work in a completely \emph{device-dependent} setting, where the state of the system and the measurement have been fully characterized, and consider two alternative formulations of this problem: a classical and a quantum one. In the classical picture, Eve can sample a random variable $\Lambda$ given the value of which there is no stochasticity in her description of the experiment. For the quantum case, we consider the model of quantum side information involving a \emph{generalized} Naimark dilation of the user's measurement, introduced by Frauchiger et al. \cite{frauchiger2013true}. In this model, Eve is allowed to have a quantum system $E$ correlated with the system being measured and with the ancillary system in the dilation.

It is a well-known result that when the measurement is assumed to be projective (or, more generally, extremal), Eve's guessing probability in the classical and quantum pictures coincide. Our first result (Theorem \ref{thm:equality-pguesses-povms}) is that this is also the case when the measurement is arbitrary but the state is pure.
Then, we move to the more relevant case in which, both, the prepared state and the implemented measurement are subject to (in general, correlated) noise, and provide a framework to estimate the produced quantum randomness. For this general scenario, we first show that Eve's guessing probability in the quantum picture is always greater than or equal to the classical one (Theorem \ref{thm:conditions-for-eq}). Our main result, however, is that there exist states and measurements for which the inequality is strict (Theorem \ref{thm:main}). In Table \ref{tab:tabla-quantum-vs-classical} we summarize the relative strengths of classical and quantum guessing probabilities for the different combinations of types of states and measurements. We finally illustrate the applicability of our approach by considering an experiment in which noisy single-photon detectors are applied to the two-mode state resulting from a single photon impinging into a balanced beam-splitter. The bounds on the guessing probability we derive demonstrate that Eve could make a more informed guess on the obtained results than when using the measurement model in~\cite{frauchiger2013true}.

\begin{table*}[t]
\centering
 \begin{tabular}{|c | c | c|} 
 \hline
  & Projective measurement & General POVM \\ [0.5ex] 
 \hline\hline
 Pure state & $p^\text{Q}_{{\rm guess}}(X|E) = p^\text{C}_{{\rm guess}}(X|\Lambda)$ & $p^\text{Q}_{{\rm guess}}(X|E) = p^\text{C}_{{\rm guess}}(X|\Lambda)$ \textbf{[This work]} \\
 General state & $p^\text{Q}_{{\rm guess}}(X|E) = p^\text{C}_{{\rm guess}}(X|\Lambda)$ & $p^\text{Q}_{{\rm guess}}(X|E) \geq p^\text{C}_{{\rm guess}}(X|\Lambda)$ \textbf{[This work]} \\
 & & $\exists \rho, \{M_S^x\}_x$ $p^\text{Q}_{{\rm guess}}(X|E,\rho,\{M_S^{x}\}_x)>p^\text{C}_{{\rm guess}}(X|\Lambda,\rho,\{M_S^{x}\}_x)$ \textbf{[This work]}\\
 \hline
 \end{tabular}
 \caption{Relationship between the classical and quantum guessing probabilities. Prior to this work, they were only known to be equivalent for projective measurements. In this work we proved (i) the equivalence for pure states and an arbitrary POVMs and (ii) that the quantum guessing probability can be strictly larger than the classical in the most general scenario.}
  \label{tab:tabla-quantum-vs-classical} 
\end{table*}


\section{Noisy preparation}
Before presenting our contributions, it is worth reviewing the known results for the setting of a projective measurement (PM) on a system in a mixed state. Let us start with a toy example. Consider that a measurement in the computational basis $\{\ket{0},\ket{1}\}$ is conducted on a qubit $S$ in the $\ket{+}=(\ket 0+\ket 1)/\sqrt 2$ state and that the outcome $+1$ is obtained. Suppose that this outcome was to be communicated to an interested user, Alice, but, before that, an eavesdropper, Eve, learns this outcome and then destroys any record of it.
If Alice, knowing that the measurement was performed but ignoring the outcome, wants to describe the state of system $S$, she has to associate to it the ensemble of states $\{p_X(x),\ket{x}\}$ with $p_X(0)=p_X(1)=1/2$ and represent it with the maximally mixed state $\frac{\mathbb{I}}{2}=\frac{1}{2}[\ket{0}\bra{0}+\ket{1}\bra{1}]$. Therefore, to the question of what would be the result of a second measurement in the $X$ basis, she can do no better than a uniformly random guess. Eve, on the other hand, having the additional \emph{classical side information} of the first measurement's outcome, has a better (in fact, complete) description of the state of $S$ and, therefore, can deterministically predict that the second outcome will be $+1$. 

What this simple example shows is that when one represents the state of a system $S$ with a mixed state $\rho_S$ compatible with an ensemble $\{p_\Lambda(\lambda),\ket{\lambda}\}$, one can never rule out the possibility of Eve being classically correlated with the state of the system via the random variable $\Lambda$. She, after learning (or, sampling) a value $\lambda$ for $\Lambda$, can make a better prediction for the outcome of a measurement on $S$ than the honest user Alice. Moreover, since for a given mixed state $\rho_S$ there are infinitely many ensembles compatible with it, to assess the unpredictability of the outcomes of a measurement one has to consider them all, as some might provide better predicting power to Eve than others. 
This rationale leads to defining Eve's \emph{classical guessing probability} as
\begin{align}\label{eq:classical-pguess-mixed-states}
& p^\text{C}_{{\rm guess}}(X|\Lambda,\rho_S,\{\Pi_S^x\}_x):= &&\nonumber\\
&\qquad\max_{p(\lambda),\ket{\varphi_\lambda}_S}\sum_{\lambda}p(\lambda) \max_x \bra{\varphi_\lambda}\Pi_S^{x}\ket{\varphi_\lambda}_S &&\nonumber\\
&\qquad~\textrm{subject to } \sum_\lambda p(\lambda)\ket{\varphi_\lambda}\bra{\varphi_\lambda}_S = \rho_S.&
\end{align}
Notice that when $\rho_S=\ket{\psi}\bra{\psi}_S$ is pure, because of its extremality in the set of states, $p^\text{C}_{{\rm guess}}(X|\Lambda,\rho_S,\{\Pi_S^x\}_x)=\max_x \bra{\psi}\Pi_S^{x}\ket{\psi}_S$ and, hence, all the observed randomness is of a quantum origin.

One could have considered that Eve, rather than having access to some classical random variable $\Lambda$, has access to another quantum system $E$ (or, to the environment) such that the global state $\rho_{SE}=\ket{\psi}\bra{\psi}_{SE}$ is (without loss of generality) pure. By the measurement of $\{\Pi_S^x\}$ on $S$, the classical outcome $x$ is produced with probability $p(x)=\Tr[\Pi_S^x\rho_S]$ and the state of the environment is steered to $\rho_E^x=\Tr_S[(\Pi_S^x\otimes \mathbb{I}_E)\ket{\psi}\bra{\psi}_{SE}]/p(x)$. 
Given that the states $\rho_{E}^x$ are, in general, not diagonal in the same basis, we say that Eve holds \emph{quantum side information} about the random variable $X$. Eve then chooses a measurement $\{M_E^x\}_x$ trying to maximise the probability that its outcome (her prediction) is $x$ when the steered state was $\rho_{E}^x$. In other words, she performs a measurement maximising $\sum_{x} p(x)\Tr[M_E^x\rho_E^x]$, the average probability to distinguish the states $\rho_E^x$ occurring with probability $p(x)$. Eve's \emph{quantum guessing probability} \cite{konig2009operational} is then given by
\begin{align}\label{eq:quantum-pguess-mixed-states}
  p^\text{Q}_{\rm guess}(X|E,&\rho_S,\{\Pi_S^x\}_x) \nonumber\\
  &:= \max_{\{M_E^{x}\}_{x}} \sum_{x} p(x)\Tr[M_E^x\rho_E^x]\nonumber\\
  &=\max_{\{M_E^{x}\}_{x}} \sum_{x}\bra{\psi}\Pi_{S}^{x}\otimes M_E^{x}\ket{\psi}_{SE}
\end{align}
where $\ket{\psi}_{SE}$ is any fixed purification of $\rho_S$ (they are all equivalent up to a unitary in $E$, which can be absorbed in the optimisation over $\{M_E^x\}_x$). 

Theorem \ref{thm:equality-pguesses-mixed-states} states the well-known result that these two different ways of quantifying Eve's predicting power are equivalent. In other words, entanglement does not provide Eve with an advantage in the task of guessing the outcomes of a PM on a mixed state.

\begin{theorem}[Folkore]\label{thm:equality-pguesses-mixed-states}
  For every state $\rho_S$ and every PM $\{\Pi_S^x\}_x$,
  \begin{equation*}
    p^\text{C}_{{\rm guess}}(X|\Lambda,\rho_S,\{\Pi_S^x\}_x)
    =p^\text{Q}_{\rm guess}(X|E,\rho_S,\{\Pi_S^x\}_x)~.
\end{equation*}
\end{theorem}
This result in fact holds for any extremal measurement, not necessarily projective.

\section{Noisy measurement}
Before studying the most general scenario, let us first consider the case in which a general measurement, represented by a Positive-Operator Valued Measure (POVM) $\{M_S^x\}_x$, is measured on a system $S$ in a pure state $\ket{\phi}_S$. Given that the set of POVMs is, just like the set of quantum states, convex \cite{d2005classical}, we can proceed via analogy with the case of a mixed state and assume that Eve can now sample a random variable $\Lambda$ such that $M_S^x=\sum_\lambda p(\lambda) M_S^{x,\lambda}$ with $\{M_S^{x,\lambda}\}_x$ POVMs for all $\lambda$. With her knowledge of $\lambda$, her best prediction for the outcome of the measurement on $S$ is $\argmax_x \bra{\phi}M_{S}^{x,\lambda}\ket{\phi}$ and this is correct with probability $\sum_\lambda p(\lambda)\max_x \bra{\phi}M_{S}^{x,\lambda}\ket{\phi}$. Finally, by letting Eve optimise over all possible convex combinations, her \emph{classical guessing probability} is
\begin{align}\label{eq:classical-pguess-povm}
& p^\text{C}_{{\rm guess}}(X|\Lambda,\ket{\phi}_S,\{M_S^x\}_x):= &&\nonumber\\
&\qquad \max_{p(\lambda),\{M_S^{x,\lambda}\}_x}\sum_\lambda p(\lambda)\max_x \bra{\phi}M_{S}^{x,\lambda}\ket{\phi}_S &&\nonumber\\
&\qquad~\textrm{subject to } \sum_\lambda p(\lambda) M_S^{x,\lambda} = M_S^x\text{ for all }x.&
\end{align}
Analogously to the case of a pure state, when $\{M_S^x\}_x$ is extremal (but not necessarily projective \cite{d2005classical}) we have completely intrinsic quantum randomness, that is $p^\text{C}_{{\rm guess}}(X|\Lambda,\ket{\phi}_S,\{M_S^x\}_x)=\max_x \bra{\phi}M_S^{x}\ket{\phi}_S$.

\medskip

A notion of a quantum guessing probability for the case of general POVMs was given in \cite{frauchiger2013true}. One assumes that what is seen as a POVM $M_S$ on system $S$ is, in fact, a PM $\{\Pi_{SA}^x\}$ on $S$ and an ancillary system $A$. In fact, $\{\Pi_{SA}^x\}_x$ is a Naimark extension of $M_S$ and the correlations with Eve are modelled via a mixed state $\sigma_A$ on $A$, of which she holds a purification $\ket{\psi_{AE}}$. See Fig. \ref{fig:quantum-side-info} for a schematic description of this model of quantum side information.  Then, as in the case studied in the previous section, Eve optimises over measurements $\{M_E^x\}_x$ trying to maximise on average the correlation $\bra{\phi,\psi}\Pi_{SA}^{x}\otimes M_E^{x}\ket{\phi,\psi}_{SAE}$. Eve's quantum guessing probability is thus given by
\begin{align}\label{eq:quantum-pguess-povm}
& p^\text{Q}_{{\rm guess}}(X|E,\ket{\phi}_S,\{M_S^x\}_x):= &&\nonumber\\
&\qquad \max_{\{\Pi_{SA}^{x}\}_{x},\ket{\psi}_{AE},\{M_E^{x}\}_{x}}\sum_{x}\bra{\phi,\psi}\Pi_{SA}^{x}\otimes M_E^{x}\ket{\phi,\psi}_{SAE} &&\nonumber\\
&\qquad~\textrm{subject to } &\nonumber\\
&\qquad\qquad \mathrm{Tr}_A[\Pi_{SA}^{x}(\mathbb{I}_S\otimes \Tr_{E}[\ket{\psi}\bra{\psi}_{AE}])]=M^x_S \text{ for all } x.&
\end{align}

\begin{figure}[h!]
     \includegraphics[width=\columnwidth]{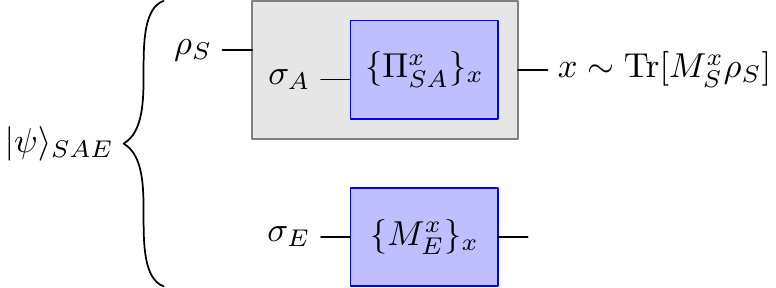}
\caption{Schematic description of our model for quantum side information. Eve chooses a projective implementation $(\{\Pi_{SA}^x\}_x,\sigma_A)$ of the user's POVM $\{M_S^x\}_x$ and she is allowed to be entangled with, both, the system $S$ and the ancilla $A$. With this, she optimises over measurements on her subsystem maximising the correlation with the user's measurement outcomes (see Eq. \eqref{eq:pguess-entangled-ancilla-general}).} \label{fig:quantum-side-info}
\end{figure}

There is an important difference between Eq. \eqref{eq:quantum-pguess-povm} and the analogous in the framework introduced in \cite{frauchiger2013true}. In the latter, the particular projective implementation $(\{\Pi_{SA}^{x}\}_x,\sigma_A)$ for a given POVM $\{M_S^x\}_x$ has to be specified by the user. In this work, we let it be chosen, in fact optimised, by Eve. This is more natural when quantifying randomness, especially in adversarial scenarios. Our first result, Theorem \ref{thm:equality-pguesses-povms}, is the analogous of Theorem \ref{thm:equality-pguesses-mixed-states}, now for noisy measurements and pure states, instead of noisy states and PMs.
\begin{theorem}\label{thm:equality-pguesses-povms}
For every pure state $\ket{\phi}_S$ and every POVM $\{M_S^x\}_x$, 
\begin{equation*}
p^\text{C}_{{\rm guess}}(X|\Lambda,\ket{\phi}_S,\{M_S^x\}_x)=p^\text{Q}_{\rm guess}(X|E,\ket{\phi}_S,\{M_S^x\}_x).
\end{equation*} 
\end{theorem}

Theorem \ref{thm:equality-pguesses-povms}, in fact, follows as a corollary of a theorem for the most general scenario, which we state in the following section.

\section{Noisy preparation and measurement}
We arrive, now, at the most general setting. Let us consider that a POVM $\{M_S^x\}_x$ is measured on a system $S$ in a state $\rho_S$. When considering classical side information, we now let Eve choose convex decompositions of, both, the state and the measurement. Her classical guessing probability is thus given by
\begin{align}\label{eq:pguess-extremals-general}
& p^\text{C}_{{\rm guess}}(X|\Lambda,\rho_S,\{M_S^x\}_x):= &&\nonumber\\
&\qquad \max_{p(i,j),\ket{\varphi_i}_S,\{M_S^{x,j}\}_{x}}\sum_{i,j}p(i,j) \max_x \bra{\varphi_i}M_S^{x,j}\ket{\varphi_i}_S &&\nonumber\\
&\qquad~\textrm{subject to } &\nonumber\\
&\qquad\qquad\qquad \sum_{i,j} p(i,j)\ket{\varphi_i}\bra{\varphi_i} = \rho_S&\nonumber\\
&\qquad\qquad\qquad \sum_{i,j} p(i,j)M_S^{x,j} = M_S^{x} \text{ for all } x&\nonumber\\
&\qquad\qquad\qquad \sum_{i,j}p(i,j)\braket{\varphi_i|M_S^{x,j}|\varphi_i}=\Tr[M_S^x\rho_S]
\end{align}
The last condition in this optimization problem states that Eve's strategy, although potentially correlating the choices of pure state and extremal measurement, cannot be arbitrary, as it must reproduce the observed statistics on $S$. In the restricted cases of the previous sections, we do not need to explicitly impose this because it follows immediately from the convex decomposition requirement.  

In the case of quantum side information, we let Eve hold a purification of the joint state $\rho_{SA}$ of the system $S$ plus the ancillary system used for her choice of a projective implementation $(\{\Pi_{SA}^x\},\sigma_A)$ of $\{M_S^x\}_x$. Notice that we do not assume that $\rho_{SA}=\rho_S\otimes \sigma_A$. Eve's quantum guessing probability is thus given by
\begin{align}\label{eq:pguess-entangled-ancilla-general}
& p^\text{Q}_{{\rm guess}}(X|E,\rho_S,\{M_S^x\}_x):= &\nonumber\\
&\qquad\max_{\{\Pi_{SA}^{x}\}_{x},\ket{\psi}_{SAE},\{M_E^{x}\}_{x}}\sum_{x}\bra{\psi}\Pi_{SA}^{x}\otimes M_E^{x}\ket{\psi}_{SAE} &&\nonumber\\
&\qquad~\textrm{subject to } &\nonumber\\
&\qquad\quad \Tr_{AE}[\ket{\psi}\bra{\psi}_{SAE}]=\rho_S &\nonumber\\
&\qquad\quad \mathrm{Tr}_A[\Pi_{SA}^{x}(\mathbb{I}_S\otimes \Tr_{SE}[\ket{\psi}\bra{\psi}_{SAE}])]=M^x_S \text{ for all } x&\nonumber\\
&\qquad\quad \braket{\psi|\Pi_{SA}^x\otimes\mathbb{I}_E|\psi}_{SAE}=\Tr[M^{x}_S\rho_S]
\end{align}

As a warm up for our main result, we first state Theorem \ref{thm:conditions-for-eq}, whose proof we defer to Appendix A. Its first part says that, as expected, any general strategy involving classical side information can be implemented in the quantum picture. Its second part is a sufficient condition for there to be an equality between the classical and the quantum guessing probabilities in this general scenario.

\begin{theorem}\label{thm:conditions-for-eq} Let $\rho_S$ be a state, $\{M_S^x\}$ a POVM and $p^\text{C}_{{\rm guess}}(X|\Lambda,\rho_S,\{M_S^x\}_x)$ and $p^\text{Q}_{{\rm guess}}(X|E,\rho_S,\{M_S^x\}_x)$ as defined in Eqs. \eqref{eq:pguess-extremals-general} and \eqref{eq:pguess-entangled-ancilla-general} respectively. Then,
\begin{enumerate}

\item $p^\text{C}_{{\rm guess}}(X|\Lambda,\rho_S,\{M_S^x\}_x)\leq p^\text{Q}_{{\rm guess}}(X|E,\rho_S,\{M_S^x\}_x)$.

\item If $p^\text{Q}_{{\rm guess}}(X|E,\rho_S,\{M_S^x\}_x)$ has an optimal solution $\langle \{\Pi_{SA}^{x}\}_{x},\ket{\psi}_{SAE},\{M_E^{x}\}_{x}\rangle$ such that the postmeasurement states on SA
\begin{align*}
\rho_{SA}^x=\frac{\Tr_E[(\mathbb{I}_{SA}\otimes M_{E}^{x})\ket{\psi}\bra{\psi}_{SAE}]}{\bra{\psi}\mathbb{I}_{SA}\otimes M_{E}^{x}\ket{\psi}_{SAE}}
\end{align*}
are all separable, then $p^\text{Q}_{{\rm guess}}(X|E,\rho_S,\{M_S^x\}_x)\leq p^\text{C}_{{\rm guess}}(X|\Lambda,\rho_S,\{M_S^x\}_x)$.
\end{enumerate}
\end{theorem}

It is straightforward to see that Theorems \ref{thm:equality-pguesses-mixed-states} and \ref{thm:equality-pguesses-povms} immediately follow as a corollaries of Theorem \ref{thm:conditions-for-eq}. For example, if $\rho_S=\ket{\phi}\bra{\phi}_S$, then the postmeasurement states on $SA$ after any measurement on $E$ are necessary separable (in fact, product), implying then, by Theorem $\ref{thm:conditions-for-eq}$, that $p^\text{C}_{{\rm guess}}(X|\Lambda,\ket{\phi}_S,\{M_S^x\}_x)=p^\text{Q}_{\rm guess}(X|E,\ket{\phi}_S,\{M_S^x\}_x)$. Same reasoning holds for Theorem $\ref{thm:equality-pguesses-mixed-states}$. 
%

From the second part of Theorem \ref{thm:conditions-for-eq} it follows that, if there is to be an advantage for Eve in the quantum scenario, it must come from her preparing an entangled state between $S$ and $A$ via her measurement. Building on this fact, our main result, Theorem \ref{thm:main} below, is the construction of a $4$-outcome qubit measurement (in fact, a family of these) for which Eve's quantum guessing probability is perfect and strictly greater than the classical one.

\begin{theorem}\label{thm:main} There exists a $4$-outcome qubit POVM $\{M_S^{x}\}_x$ such that $$1=p^\text{Q}_{{\rm guess}}(X|E,\frac{\mathbb{I}_S}{2},\{M_S^{x}\}_x)>p^\text{C}_{{\rm guess}}(X|\Lambda,\frac{\mathbb{I}_S}{2},\{M_S^{x}\}_x).$$
\end{theorem}

\begin{proof}[Proof sketch.] The proof of Theorem \ref{thm:main} can be found in Appendix B. Here we sketch the main parts. Let $\{\ket{\Phi_x^\theta}\}_{x=1}^4$ be the parametric family, indexed by $\theta\in[0,\pi/2]$, of entangled bases for a space of two qubits defined in \cite[Eq. (3)]{tavakoli2021bilocal}. We set
\begin{align*}
\{\Pi_{SA}^{x,\theta}\}_x=\{\ket{\Phi_x^\theta}\bra{\Phi_x^\theta}\}_x\qquad &\text{and   }&\rho_{SA}=\frac{\mathbb{I}_{SA}}{4},
\end{align*}
and, therefore, 
\begin{align*}
\{M_S^{x,\theta}\}_x=\{\Tr_A[\Pi_{SA}^{x,\theta}\cdot \frac{\mathbb{I}_{SA}}{2}]\}_x\quad &\text{and   }&\rho_S=\frac{\mathbb{I}_{S}}{2}.
\end{align*}

It is straightforward to see that Eve can achieve 
$p^\text{Q}_{{\rm guess}}(X|E,\rho_S,\{M_S^{x,\theta}\}_x)=1$ if she steers the ensemble $\{1/4,\ket{\Phi_x^\theta}\}_x$ on $SA$ by measuring her share of $\ket{\psi}_{SAE}=\sum_x 1/2 \ket{\Phi_x^\theta}_{SA}\ket{x}_E$ in the $\{\ket{x}\}$ basis.
As for the classical guessing probability being strictly below $1$, this follows from two technical results which we prove in the appendix. The first one states that for the settings in which extremal measurements are necessarily rank-one (e.g., $d^2$-outcome measurements, where~$d$ is the dimension of~$S$), having $p^\text{C}_{{\rm guess}}(X|\Lambda,\rho,\{M_S^{x}\}_x)=1$ implies that $\{M_S^{x}\}_x$ is a convex combination of PMs. In \cite[Eq. (56)]{oszmaniec2017simulating}, the class of $4$-outcome qubit POVMs which are convex combinations of PMs was shown to be definable with an semidefinite program (SDP). We numerically checked that for values of $\theta\in[0,\pi/10]$, the POVMs $\{M_S^{x,\theta}\}_x$ are not a convex combination of PMs \cite{pythoncode}. To end the proof, and for concreteness, we set $\{M_S^{x}\}_x=\{M_S^{x,\theta}\}_x$ for $\theta=0$.
\end{proof}

We conclude this section with a second corollary to Theorem \ref{thm:conditions-for-eq}, which applies to a restricted adversarial setting. Consider the case in which the quantum adversary Eve is restricted to having two separate systems $E_1$ and $E_2$, one purifying $\rho_S$ and the other one purifying $\sigma_A$, which she cannot measure jointly. From Theorem \ref{thm:conditions-for-eq} it follows that
\begin{corollary}\label{cor:restricted-adversary} Let $\tilde{p}^\text{Q}_{{\rm guess}}(X|E,\rho_S,\{M_S^x\}_x)$ be as in Eq. \eqref{eq:pguess-entangled-ancilla-general} with the additional restriction that $\ket{\psi}_{SAE}=\ket{\psi_1}_{SE_1}\ket{\psi_2}_{AE_2}$ and $M_E^x=M_{E_1}^x\otimes M_{E_2}^x$. Then,
$$\tilde{p}^\text{Q}_{{\rm guess}}(X|E,\rho_S,\{M_S^x\}_x)\leq p^\text{C}_{{\rm guess}}(X|\Lambda,\rho_S,\{M_S^x\}_x).$$
\end{corollary}

\section{Application to a QRNG}
In \cite[Examples 1-3]{frauchiger2013true}, the following simple model of an imperfect QRNG based on a beam splitter (BS) with inefficient detectors is considered. Let $\ket{\psi}_{12}:=\frac{1}{\sqrt{2}}(\ket{10}+\ket{01})$ be the two-mode state obtained after sending a single photon onto an ideal BS. Let there be detectors with efficiency $\mu\in [0,1]$ at each of two output paths of the BS and let $M_D^1=\mu\ket{1}\bra{1}_D$ be the operator of a two-outcome POVM $\{M_D^0,M_D^1\}$ representing the detection of $1$ photon on path $D\in\{1,2\}$. If we measure each path separately, the overall measurement can be represented by the POVM
\begin{align*}
M_\mu=\{M_{1}^0\otimes M_{2}^0,M_{1}^0\otimes M_{2}^1,M_{1}^1\otimes M_{2}^0,M_{1}^1\otimes M_{2}^1\}.
\end{align*}
As we noted before, in order to use the framework in \cite{konig2009operational} to quantify the unpredictability of this QRNG's outcomes one has to decide on a particular projective implementation of $\{M_S^x\}$. In \cite[Example 3]{frauchiger2013true}, the authors consider the projective implementation $(\{\Pi_{11'}^x\otimes \Pi_{22'}^y\}_{x,y},\sigma_{1'}\otimes \sigma_{2'})$ of $M_\mu$ with
\begin{align}\label{eq:projective-implementation-frauchiger}
\Pi_{DD'}^1&=\ket{1}\bra{1}_{D}\otimes \ket{1}\bra{1}_{D'}\text{ and }\nonumber\\
\sigma_{D'}&=(1-\mu)\ket{0}\bra{0}_{D'}+\mu \ket{1}\bra{1}_{D'}.
\end{align}
In Fig. 1, we plot Eve's guessing probability for this particular projective implementation
\begin{align}
f(\mu):=\max_{\{M_E^{x,y}\}_{x,y}}\sum_{x,y}\bra{\psi}\Pi_{11'}^x\otimes \Pi_{22'}^y\otimes M_E^{x,y}\ket{\psi}_{11'22'E}
\end{align}
with $\ket{\psi}_{11'22'E}$ a fixed purification of $\ket{\psi}\bra{\psi}_{12}\otimes \sigma_{1'}\otimes \sigma_{2'}$ and compare it to $p^\text{Q}_{{\rm guess}}(X|E,\ket{\psi}_{12},M_\mu)$, as a function of the (decreasing) efficiency $\mu$ of the detectors \cite{pythoncode}.

\begin{figure}[h!]
\includegraphics[scale=0.6]{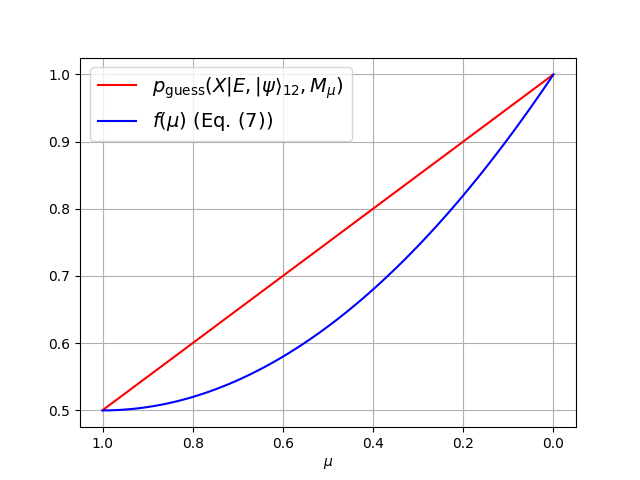}
\caption{Comparison between our work and \cite{frauchiger2013true}. In this plot we show that if, rather than assuming the particular projective implementation of inefficient detectors on the outputs of an ideal BS used in \cite{frauchiger2013true} and reproduced in Eq. \ref{eq:projective-implementation-frauchiger}, we let it chosen by the Eve, her guessing probability is strictly bigger for every value of the efficiency $\mu\in(0,1)$.\label{fig:pguess-comp}}
\end{figure}

We see that, for every value of the efficiency $\mu\in(0,1)$, fixing the particular projective extension in Eq. \eqref{eq:projective-implementation-frauchiger} strictly decreases Eve's maximal guessing probability. In other words, the projective implementation in Eq. \eqref{eq:projective-implementation-frauchiger} leads to an underestimation of Eve's guessing probability for every $\mu\in(0,1)$.

\section{Conclusions}
In this work, we have studied the unpredictability of the outcomes of a general quantum measurement from the point of view of an eavesdropper holding side information correlated to, both, the state of the system and the measurement. 
We have shown that, while the quantum and classical guessing probabilities coincide in the case of extremal states (i.e. pure) or extremal measurements, equivalence does not hold in general.

The classical and quantum guessing probabilities not coinciding in the general scenario has immediate consequences for the design of device-dependent QRNGs, as proper justifications should be issued regarding which of the two pictures is assumed. However, being the characterization of intrinsic randomness quite a general quantum mechanical problem, our result might find applicability beyond QRNGs.

As for future research directions, it would be interesting to characterize the set of states and measurements for which there is a quantum advantage in the guessing probability. Last but not least, from a practical perspective, it is important to come up with ways to compute (or, at least, computably approximate from above) these quantities.

\medskip

\emph{Note added:} while completing this work we became aware of a recent work~\cite{newpaper} in which a similar approach to characterise the intrinsic randomness under quantum measurements is introduced.

\medskip

\begin{acknowledgments}
\emph{Acknowledgments.} We thank M\'at\'e Farkas for fruitful discussions. We acknowledge financial support from the ERC AdG CERQUTE, the EU project QRANGE (Grant No. 820405), the AXA Chair in Quantum Information Science, the Government of Spain (FIS2020-TRANQI, NextGen Recovery Funds and Severo Ochoa CEX2019-000910-S), Fundaci\'o Cellex, Fundaci\'o Mir-Puig and Generalitat de Catalunya (CERCA, AGAUR SGR 1381).
\end{acknowledgments}

\pagebreak

\begin{widetext}

\appendix

\section{Appendix A: Proof of Theorem \ref{thm:conditions-for-eq}}
We will prove each part of Theorem \ref{thm:conditions-for-eq} in two different lemmas.

\begin{lemma} $p^\text{C}_{{\rm guess}}(X|\Lambda,\rho_S,\{M_S^x\}_x)\leq p^\text{Q}_{{\rm guess}}(X|E,\rho_S,\{M_S^x\}_x)$.
\end{lemma}

\begin{proof}
Let $\{(p(i,j),\{M_S^{x,j}\}_{x},\ket{\varphi_i})\}$ be an optimal solution to $p^\text{C}_{{\rm guess}}(X|\Lambda,\rho_S,\{M_S^x\}_x)$. Consider a bipartite ancillary system $A=A_1A_2$ initially in the state 
$$\sigma_{A}=\ket{0}\bra{0}_{A_1}\otimes\sum_{j} p(j) \ket{j}\bra{j}_{A_2}$$ 
with $\dim(\mathcal{H}_{A_1})=|X|$. Let Eve hold the purification
$$\ket{\psi}_{SAE}=\sum_{i,j} \sqrt{p(i,j)}\ket{\varphi_{i}}_S\ket{0,j}_{A_1A_2}\ket{i,j}_E.$$
Notice that $\Tr_{AE}[\ket{\psi}\bra{\psi}]=\rho_S$ and $\Tr_{SE}[\ket{\psi}\bra{\psi}]=\sigma_A$. We define the action of an operator $\tilde{U}$ as
$$\tilde{U}\ket{\phi,0,j}_{SA_1A_2}=\sum_x \sqrt{M_S^{x,j}}\ket{\phi,x,j}_{SA_1A_2},$$
for arbitrary $\ket{\phi}_S\in\mathcal{H}_S$ and let $U$ be any unitary extension of $\tilde{U}$ to the whole space $\mathcal{H}_{S}\otimes\mathcal{H}_{A_1}\otimes\mathcal{H}_{A_2}$. Let $$\Pi_{SA}^x=U^\dagger(\mathbb{I}_S\otimes \ket{x}\bra{x}\otimes \mathbb{I}_{A_2})U.$$ 
Notice that, for a given $\ket{\phi}_S\in\mathcal{H}_S$,
\begin{align*}
\bra{\phi,0,j}\Pi_{SA}^x\ket{\phi,0,j}&=\bra{\phi,0,j}U^\dagger(\mathbb{I}_S\otimes \ket{x}\bra{x}\otimes \mathbb{I}_{A_2})U \ket{\phi,0,j}\\
&=\sum_{x',x''}\bra{\phi}\sqrt{M^{x',j}_S}^\dagger\bra{x',j}(\mathbb{I}_S\otimes \ket{x}\bra{x}\otimes \mathbb{I}_{A_2})\sqrt{M^{x'',j}_S}\ket{\phi}\ket{x'',j}\\
&=\bra{\phi}M^{x,j}_S\ket{\phi}
\end{align*}
and so
\begin{align*}
\Tr[\Pi_{SA}^x \ket{\phi}\bra{\phi}_S\otimes\sigma_A]&=\sum_{i,j}p(i,j)\bra{\phi,0,j}\Pi_{SA}^x\ket{\phi,0,j}\\
&=\sum_{i,j} p(i,j)\bra{\phi}M^{x,j}_S\ket{\phi}=\bra{\phi}M^{x}_S\ket{\phi}.
\end{align*}
Therefore, $(\{\Pi_{SA}^x\}_x,\sigma_A)$ gives a projective extension of $\{M_S^x\}_x$. Next, we define the measurement $\{M_E^x\}_x$ on Eve's subsystem. $M_E^x$ will project onto the subspace spanned by the states $\ket{i,j}$ such that $x$ maximizes the Born rule for the $i$-th state and $j$-th POVM in the decompositions of $\rho_S$ and $\{M_S^x\}_x$ respectively. That is, 
 \begin{align*}
M_E^x&=\sum_{(i,j)\in A_x}\ket{i,j}\bra{i,j}
\end{align*}
with
\begin{align}\label{eq:set-ax}
A_x&=\{(i,j)\mid x=\min \{y|\braket{\varphi_i|M^{y,j}_S|\varphi_i}=\max_z \braket{\varphi_i|M^{z,j}_S|\varphi_i}\}\}.
\end{align}
The minimisation in Eq. \eqref{eq:set-ax} is there because we want that any given $(i,j)$ is included in one and only one $A_x$ (a priori, for a given $(i,j)$ there could be distinct $x_1$ and $x_2$ such that $\braket{\varphi_i|M^{x_1,j}_S|\varphi_i}=\braket{\varphi_i|M^{x_2,j}_S|\varphi_i}$ is maximum).

\medskip

With this, we have that $\langle \{\Pi^x_{SA}\}_x,\ket{\psi}_{SAE},\{M^x_E\}_x\rangle$ is a solution to \eqref{eq:pguess-entangled-ancilla-general} with value
\begin{align*}
\sum_x \bra{\psi}\Pi_{SA}^x\otimes M_E^x\ket{\psi}_{SAE}&=\sum_x \sum_{(i,j)\in A_x} \bra{\psi}\Pi_{SA}^x\otimes \ket{i,j}\bra{i,j}\ket{\psi}_{SAE}\\
&=\sum_x \sum_{(i,j)\in A_x} p(i,j)\braket{\varphi_i|M_S^{x,j}|\varphi_i}\\
&=\sum_{i,j}p(i,j) \max_x \bra{\varphi_i}M_S^{x,j}\ket{\varphi_i}_S\\
&= p^\text{C}_{{\rm guess}}(X|\lambda,\rho_S,\{M_S^x\}_x).
\end{align*}
Therefore, 
\begin{align*}
p^\text{Q}_{{\rm guess}}(X|E,\rho_S,\{M_S^x\}_x)&\geq \sum_x \bra{\psi}\Pi_{SA}^x\otimes M_E^x\ket{\psi}_{SAE} \\
&\geq p^\text{C}_{{\rm guess}}(X|\lambda,\rho_S,\{M_S^x\}_x).
\end{align*}
\end{proof}
\medskip

\begin{lemma} If $p^\text{Q}_{{\rm guess}}(X|E,\rho_S,\{M_S^x\}_x)$ has an optimal solution $\langle \{\Pi_{SA}^{x}\}_{x},\ket{\psi}_{SAE},\{M_E^{x}\}_{x}\rangle$ such that the postmeasurement states on SA
\begin{align*}
\rho_{SA}^x=\frac{\Tr_E[(\mathbb{I}_{SA}\otimes M_{E}^{x})\ket{\psi}\bra{\psi}_{SAE}]}{\bra{\psi}\mathbb{I}_{SA}\otimes M_{E}^{x}\ket{\psi}_{SAE}}
\end{align*}
are all separable, then $p^\text{Q}_{{\rm guess}}(X|E,\rho_S,\{M_S^x\}_x)\leq p^\text{C}_{{\rm guess}}(X|\Lambda,\rho_S,\{M_S^x\}_x)$.
\end{lemma}

\begin{proof}
Let $\langle \{\Pi_{SA}^{x}\}_{x},\ket{\psi}_{SAE},\{M_E^{x}\}_{x}\rangle$ be an optimal solution to $p^\text{Q}_{{\rm guess}}(X|E,\rho_S,\{M_S^x\}_x)$ and let
$$\tau_{SA}^x=\frac{\Tr_E[(\mathbb{I}_S\otimes\mathbb{I}_A\otimes M_{E}^{x})\ket{\psi}\bra{\psi}_{SAE}]}{p(x)}$$
with
$$p(x)=\Tr[(\mathbb{I}_S\otimes\mathbb{I}_A\otimes M_{E}^{x})\ket{\psi}\bra{\psi}_{SAE}].$$
Notice that,
\begin{align*}
p^\text{Q}_{{\rm guess}}(X|E,\rho_S,\{M_S^x\}_x)=\sum_x p(x)\Tr[\Pi_{SA}^x \tau_{SA}^x].
\end{align*}

Let us now assume that the states $\tau_{SA}^x$ are separable. Then, 

\begin{align*}
\tau_{SA}^x&=\sum_{i} p(i|x) \ket{\varphi_i^x}\bra{\varphi_i^x}_S\otimes \ket{\phi_i^x}\bra{\phi_i^x}_A\\
\end{align*}

We define 
\begin{align*}
\rho_{S|x,i}&:=\ket{\varphi_i^x}\bra{\varphi_i^x}_S\\
F^{y}_{S|x,i}&:=\Tr_A[\Pi_{SA}^y(\mathbb{I}_S\otimes \ket{\phi_i^x}\bra{\phi_i^x}_A)]
\end{align*}
Notice that,
\begin{align*}
\sum_{x,i}p(x,i)\rho_{S|x,i}&=\sum_x p(x)\sum_i p(i|x)\ket{\varphi_i^x}\bra{\varphi_i^x}=\sum_{x}p(x)\rho_{S}^x=\rho_S\\
\sum_{x,i}p(x,i)F^{y}_{S|x,i}&=\sum_x p(x)\sum_i p(i|x)\Tr_A[\Pi_{SA}^y(\mathbb{I}_S\otimes \ket{\phi_i^x}\bra{\phi_i^x}_A)]\\
&=\Tr_A[\Pi_{SA}^y(\mathbb{I}_S\otimes \sum_{x} p(x)\sum_i p(i|x) \ket{\phi_i^x}\bra{\phi_i^x}_A)]=\Tr_A[\Pi_{SA}^y(\mathbb{I}_S\otimes \sigma_A)]=M^y_{S}.
\end{align*}
and that $F^{y}_{S|x,j}\succcurlyeq 0$ for all $y,x$ and $i$ and satisfies
\begin{align*}
\sum_y F^{y}_{S|x',j}&=\sum_y \Tr_A[\Pi_{SA}^{y}(\mathbb{I}_S\otimes \ket{\phi_i^x}\bra{\phi_i^x}_A)]\\
&=\Tr_A[(\sum_y \Pi_{SA}^{y})(\mathbb{I}_S\otimes \ket{\phi_i^x}\bra{\phi_i^x}_A)]\\
&=\Tr_A[(\mathbb{I}_S\otimes \mathbb{I}_A)(\mathbb{I}_S\otimes \ket{\phi_i^x}\bra{\phi_i^x}_A)]\\
&=\mathbb{I}_S.
\end{align*}
We make the following change of variables:
\begin{align*}
z&\to (x,i).
\end{align*}
Therefore, $\{(p(z),\rho_{S|z},\{F^{y}_{S|z}\}_y)\}$ is a solution to Eq. \eqref{eq:pguess-extremals-general} with value
\begin{align*}
\sum_{z} p(z)\max_y\Tr[F^y_{S|z} \rho_{S|z}]&= \sum_{x,i} p(x,i)\max_y\Tr[F^y_{S|x,i} \rho_{S|x,i}]\\
&= \sum_{x,i} p(x,i)\max_y\Tr_S[\Tr_A[\Pi_{SA}^{y}(\mathbb{I}_S\otimes \ket{\phi_i^x}\bra{\phi_i^x}_A)]\ket{\varphi_i^x}\bra{\varphi_i^x}_S]\\
&\geq \sum_{x,i} p(x,i)\Tr_S[\Tr_A[\Pi_{SA}^{x}(\mathbb{I}_S\otimes \ket{\phi_i^x}\bra{\phi_i^x}_A)]\ket{\varphi_i^x}\bra{\varphi_i^x}_S]\\
&=\sum_{x} p(x) \sum_i p(i|x) \Tr[\Pi_{SA}^{x}(\ket{\varphi_i^x}\bra{\varphi_i^x}_S\otimes \ket{\phi_i^x}\bra{\phi_i^x}_A)]\\
&=\sum_{x} p(x) \Tr[\Pi_{SA}^{x}\tau_{SA}^x]\\
&= p^\text{Q}_{{\rm guess}}(X|E,\rho_S,\{M_S^x\}_x)
\end{align*}

\end{proof}

\paragraph{Proof of Theorem \ref{thm:conditions-for-eq}} Immediate from the above two lemmas.

\section{Appendix B: Proof of the main result (Theorem \ref{thm:main})}
In this section we will show that there exists a qubit state $\rho_S$ and a $4$-outcome POVM $\{M_S^x\}_x$ such that $p^\text{Q}_{{\rm guess}}(X|E,\rho_S,\{M_S^x\}_x)>p^\text{C}_{{\rm guess}}(X|\Lambda,\rho_S,\{M_S^x\}_x)$.

\medskip

To simplify the notation, in this section we will work with the following equivalent definition of the classical guessing probability, where we have condensed the indexes $i$ and $j$ for the decompositions of the state and the measurement into a single index $\lambda$:

\begin{align}
\label{eq:pguess-extremals-general-bis}
&~p^\text{C}_{{\rm guess}}(X|\Lambda,\rho_S,\{M_S^x\}_x):= \max_{p(\lambda),\ket{\varphi_\lambda}_S,\{M_S^{x,\lambda}\}_{x}}\sum_{\lambda}p(\lambda) \max_x \bra{\varphi_\lambda}M_S^{x,\lambda}\ket{\varphi_\lambda}_S &&\nonumber\\
&\qquad~\textrm{subject to } &\nonumber\\
&\qquad\qquad \sum_{\lambda} p(\lambda)\ket{\varphi_\lambda}\bra{\varphi_\lambda}_S = \rho_S&\nonumber\\
&\qquad\qquad \sum_{\lambda} p(\lambda)M_S^{x,\lambda} = M_S^{x} \text{ for all } x&\nonumber\\
&\qquad\qquad \sum_{\lambda}p(\lambda)\braket{\varphi_\lambda|M_S^{x,\lambda}|\varphi_\lambda}_S=\Tr[M_S^x\rho_S]\text{ for all } x
\end{align}

We begin by proving some small technical lemmas. The first one is a construction of states and measurements for which $p^\text{Q}_{{\rm guess}}(X|E,\rho_S,\{M_S^x\}_x)=1$.

\begin{lemma}[State and measurement perfect for quantum Eve]\label{lemma:state-and-measurement-trivial-for-quantum} Let $\{\ket{\psi^x}_{SA}\}_x$ be a basis for $\mathcal{H}_S\otimes \mathcal{H}_A$, with $d={\rm dim}(\mathcal{H}_S)={\rm dim}(\mathcal{H}_A)$, and let
\begin{align*}
\rho_{S}&:=\Tr_A [\frac{1}{d^2}\sum_x  \ket{\psi^x}\bra{\psi^x}_{SA}^x]\\
&=\mathbb{I}/d\\
M_S^x&:=\Tr_A[\ket{\psi^x}\bra{\psi^x}_{SA}(\mathbb{I}_S\otimes \Tr_S [\frac{1}{d^2}\sum_x  \ket{\psi^x}\bra{\psi^x}_{SA}^x])]\\
&=\Tr_A[\ket{\psi^x}\bra{\psi^x}_{SA}]/d
\end{align*}
Then, $p^\text{Q}_{{\rm guess}}(X|E,\rho_S,\{M_S^x\}_x)=1$.
\end{lemma}

\begin{proof}
Let $\ket{\psi}_{SAE}=\sum_x \frac{1}{d}\ket{\psi^x}_{SA}\ket{x}_E$, $\{\Pi_{SA}^x\}_x=\{\ket{\psi^x}\bra{\psi^x}_{SA}\}_x$ and $\{M_E^x\}_{x}=\{\ket{x}\bra{x}\}_x$. By construction, $\langle \ket{\psi}_{SAE},\{\Pi_{SA}^x\}_x,\{M_E^x\}_x\rangle$ satisfies the first two conditions the definition of $p^\text{Q}_{{\rm guess}}(X|E,\rho_S,\{M_S^x\}_x)$ in Eq. \eqref{eq:pguess-entangled-ancilla-general}. To see that it also satisfies the last one, notice that $\rho_{SA}=\rho_S\otimes\rho_A$ and hence
$\Tr[\Pi_{SA}^x\rho_{SA}]=\Tr_S[\rho_S\Tr_A[\Pi_{SA}^x(\mathbb{I}\otimes\rho_A)]=\Tr[M_S^x\rho_S]$. Finally, notice that
$\bra{\psi}\Pi_{SA}^{x}\otimes M_E^{x}\ket{\psi}_{SAE}=\bra{\psi}\ket{\psi^x}\bra{\psi^x}\otimes \ket{x}\bra{x}\ket{\psi}_{SAE}=1/d^2$ and, therefore, $\langle \ket{\psi}_{SAE},\{\Pi_{SA}^x\}_x,\{M_E^x\}_{x}\rangle$ achieves $p^\text{Q}_{{\rm guess}}(X|E,\rho_S,\{M_S^x\}_x)=1$.
\end{proof}

In the definition of the classical guessing probability $p^\text{C}_{{\rm guess}}(X|\Lambda,\rho_S,\{M_S^x\}_x)$ we assume that the states in the decomposition of $\rho_S$ are pure and, although not explicitly, that the POVMs in the decomposition of $\{M_S^x\}_x$ are extremal. In the following lemma we prove that this is without loss of generality.

\begin{lemma}[Classical maximum is achieved at extremals]\label{lemma:solution-made-of-extremals} Let

\begin{align}
&~\tilde{p}^\text{C}_{{\rm guess}}(X|\Lambda,\rho_S,\{M_S^x\}_x):= \max_{p(\lambda),\rho^\lambda_S,\{M_S^{x,\lambda}\}_{x}}\sum_{\lambda}p(\lambda) \max_x \Tr[M_S^{x,\lambda}\rho^\lambda_S] &&\nonumber\\
&\qquad~\textrm{subject to } &\nonumber\\
&\qquad\qquad \sum_{\lambda} p(\lambda)\rho^\lambda_S = \rho_S&\nonumber\\
&\qquad\qquad \sum_{\lambda} p(\lambda)M_S^{x,\lambda} = M_S^{x} \text{ for all } x&\nonumber\\
&\qquad\qquad \sum_{\lambda}p(\lambda)\Tr[M_S^{x,\lambda}\rho^\lambda_S]=\Tr[M_S^x\rho_S]\text{ for all } x
\end{align}

Then, for every $\rho_S$ and every $\{M_S^x\}_x$ there is an optimal solution $\langle p(\lambda),\rho_S^\lambda,\{M_S^{x,\lambda}\}_x\rangle$ to $\tilde{p}^\text{C}_{{\rm guess}}(X|\Lambda,\rho_S,\{M_S^x\}_x)$ with $\rho_S^\lambda$ pure and $\{M_S^{x,\lambda}\}_x$ extremal for all $\lambda$.
\end{lemma}
\begin{proof}
Let $\langle p(\lambda),\rho_S^\lambda,\{M_S^{x,\lambda}\}_x\rangle$ be an optimal solution to $\tilde{p}^\text{C}_{{\rm guess}}(X|\Lambda,\rho_S,\{M_S^x\}_x)$ and suppose there exists $\tilde{\lambda}$ such that $\{M_S^{x,\tilde{\lambda}}\}_x$ is not extremal. Then  $M_S^{x,\tilde{\lambda}}=\sum_{j}p(j)N_S^{x,\tilde{\lambda},j}$ with $\{N_S^{x,\tilde{\lambda},j}\}_x$ extremal POVMs for all $j$. Define
\begin{align*}
p(\lambda,j)&=\begin{cases}
p(\lambda) & \text{ if } \lambda\neq \tilde{\lambda}\\
p(\tilde{\lambda})p(j) & \text{otherwise}
\end{cases}\\
\rho_{S}^{\lambda,j}&=\rho_S^{\lambda}\\
M_S^{x,\lambda,j}&=\begin{cases}
M_S^{x,\lambda} & \text{ if } \lambda\neq \tilde{\lambda}\\
N_S^{x,\tilde{\lambda},j} & \text{otherwise}
\end{cases}\\.
\end{align*}
It is straightforward to see that $\langle p(\lambda,j),\rho_{S}^{\lambda,j},\{M_S^{x,\lambda,j}\}_x\rangle$ is a solution to $\tilde{p}^\text{C}_{{\rm guess}}(X|\Lambda,\rho_S,\{M_S^x\}_x)$. Finally, notice that
\begin{align*}
\sum_{\lambda,j}p(\lambda,j)\max_{x} \Tr[M_S^{x,\lambda,j}\rho_S^{\lambda,j}]&=\sum_{\lambda\neq \tilde{\lambda}} p(\lambda)\max_{x} \Tr[M_S^{x,\lambda}\rho_S^{\lambda}]+p(\tilde{\lambda})\sum_j p(j)\max_{x} \Tr[N_S^{x,\tilde{\lambda},j}\rho_S^{\tilde{\lambda}}]\\
&\geq \sum_{\lambda\neq \tilde{\lambda}} p(\lambda)\max_{x} \Tr[M_S^{x,\lambda}\rho_S^{\lambda}]+p(\tilde{\lambda})\max_x \sum_j p(j)\Tr[N_S^{x,\tilde{\lambda},j}\rho_S^{\tilde{\lambda}}]\\
&=\sum_{\lambda\neq \tilde{\lambda}} p(\lambda)\max_{x} \Tr[M_S^{x,\lambda}\rho_S^{\lambda}]+p(\tilde{\lambda})\max_{x} \Tr[M_S^{x,\tilde{\lambda}}\rho_S^{\tilde{\lambda}}]\\
&=\tilde{p}^\text{C}_{{\rm guess}}(X|\Lambda,\rho_S,\{M_S^x\}_x),
\end{align*}
with the inequality following from the convexity of the $\max$ function. Therefore $\langle p(\lambda,j),\rho_{S}^{\lambda,j},\{M_S^{x,\lambda,j}\}_x\rangle$ is also an optimal solution but with the nonextremal POVM $\{M_S^{x,\tilde{\lambda}}\}_x$ replaced by extremal POVMs. By repeating the same process with every nonextremal POVM we end up with a optimal solution whose decomposition of $\{M_S^{x}\}_S$ is just given in terms of extremal POVMs. The same reasoning holds for the nonpure states.
\end{proof}

\bigskip

The next lemma identifies a necessary condition for having $p^\text{C}_{{\rm guess}}(X|\Lambda,\rho_S,\{M_S^x\}_x)=1$ in two scenarios.

\begin{lemma}[Necessary condition for perfect classical pguess]\label{lemma:classical-pguess-max-with-projectors} In the case of $d={\rm dim}(\mathcal{H}_S)=2$ or $\{M_S^x\}_x$ having $d^2$ outcomes, if $p^\text{C}_{{\rm guess}}(X|\Lambda,\rho_S,\{M_S^x\}_x)=1$ then $\{M_S^x\}_x$ is a convex combination of PMs.
\end{lemma}

\begin{proof}
Let $\langle p(\lambda),\rho_S^\lambda,\{M_S^{x,\lambda}\}_x\rangle$ be a solution to $p^\text{C}_{{\rm guess}}(X|\Lambda,\rho_S,\{M_S^x\}_x)$ with $\sum_\lambda p(\lambda)\max_x\bra{\varphi_\lambda}M_S^{x,\lambda}\ket{\varphi_\lambda}_S=1$. Then, for every $\lambda$, 
\begin{align}\label{eq:max-classical-pguess}
\max_x\bra{\varphi_\lambda}M_S^{x,\lambda}\ket{\varphi_\lambda}_S=1.
\end{align}
We will prove that, for every $\lambda$, $\{M_S^{x,\lambda}\}_x$ is a PM. If $M_x^\lambda=\mathbb{I}$, there is nothing to prove. If not, then, either in the cases of $S$ being a qubit or in the case of $\{M_S^x\}$ having $d^2$ outcomes, the extremality of $\{M_S^{x,\lambda}\}_x$ (c.f. Lemma \ref{lemma:solution-made-of-extremals}) implies that it has to be rank one (see \cite[Corollary 2]{d2005classical}), that is $M_S^{x,\lambda}=\alpha_x \ket{\alpha_x}\bra{\alpha_x}$ with $\alpha_x\in[0,1]$ for all $x$. Therefore, from Eq. \eqref{eq:max-classical-pguess}, it follows that there exists $x'$ such that $M_S^{x',\lambda}=\ket{\varphi_\lambda}\bra{\varphi_\lambda}$. We will show that there exists $x''\neq x'$ such that $M_S^{x'',\lambda}=\mathbb{I}-\ket{\varphi_\lambda}\bra{\varphi_\lambda}$ and, therefore, $M_S^{x,\lambda}=\mathbf{0}$ for all $x\neq x',x''$.

Let $\ket{\alpha_x} = a_x \ket{\varphi_\lambda}+b_x \ket{\varphi_\lambda^\bot}$. Then,
\begin{align*}
0=|\braket{\varphi_\lambda|\varphi_\lambda^\bot}|^2&=\bra{\varphi_\lambda}\mathbb{I}-\ket{\varphi_\lambda}\bra{\varphi_\lambda}\ket{\varphi_\lambda}\\
&=\bra{\varphi_\lambda}\sum_{x\neq x'}M_S^{x,\lambda}\ket{\varphi_\lambda}\\
&=\sum_{x\neq x'} \alpha_x |{a_x}|^2.
\end{align*}
Since this is a sum of positive numbers, it implies that if $\alpha_x\neq 0$ then $a_x=0$ and, so, $\ket{\alpha_x}=\ket{\varphi_\lambda^\bot}$. Since $\sum_{x\neq x'}\alpha_x\ket{\alpha_x}\bra{\alpha_x}=\ket{\varphi_\lambda^\bot}\bra{\varphi_\lambda^\bot}$, there must exist some $x_{\bot}\neq x'$ such that $\alpha_{x_\bot}\neq 0$. Moreover, since by Lemma \ref{lemma:solution-made-of-extremals} we can assume that $\{M_S^{x,\lambda}\}_x$ is extremal, this $x_\bot$ has to be unique and equal to $1$, because otherwise there will be two different effects proportional to $\ket{\varphi_\lambda^\bot}\bra{\varphi_\lambda^\bot}$ and, hence, with coinciding support, something that cannot happen in an extremal POVM (see \cite[Corollary 3]{d2005classical}).
\end{proof}

\bigskip

For the last part of the proof, we will use the following semidefinite programming characterization of the set of four outcome qubit POVMs which are convex combinations of PMs given in \cite{oszmaniec2017simulating}.

\begin{lemma}{\cite[Supplemental material, Eq. (56)]{oszmaniec2017simulating}}\label{lemma:pm-simulability}
A four outcome qubit POVM $\{M_x\}_{x=1}^4$ is a convex combination of PMs if and only if 
	\begin{align*}\
	M_{1}&=N_{12}^{+}+N_{13}^{+}+N_{14}^{+}\\
	M_{2}&=N_{12}^{-}+N_{23}^{+}+N_{24}^{+}\\
	M_{3}&=N_{13}^{-}+N_{23}^{-}+N_{34}^{+}\\
	M_{4}&=N_{14}^{-}+N_{24}^{-}+N_{34}^{-}
	\end{align*}
	where Hermitian operators $N_{ij}^{\pm}$ satisfy $N_{ij}^{\pm}\geq0$
	and $N_{ij}^{+}+N_{ij}^{-}=p_{ij}\mathbb{I}$ for $i<j$ , $i,j=1,2,3,4$, where $p_{ij}\geq0$ and $\sum_{i<j}p_{ij}=1$.
\end{lemma}

%

\bigskip

We now have all the ingredients ready for the proof of our main theorem:

\begin{proof}[Proof of Theorem \ref{thm:main}]
First, let us introduce the family, parametrized by $\theta\in[0,\pi/2]$, of entangled basis $\{\ket{\Phi_x^\theta}\}_{x=1}^4$ for a space of two qubits defined in \cite[Eq. (3)]{tavakoli2021bilocal}

\begin{equation}
\ket{\Phi_b^\theta}:=\frac{\sqrt{3}+e^{i\theta}}{2\sqrt{2}}\ket{\vec{m}_b,{-}\vec{m}_b}+\frac{\sqrt{3}-e^{i\theta}}{2\sqrt{2}}\ket{{-}\vec{m}_b,\vec{m}_b}.
\end{equation}

with

$$\ket{{\pm}\vec{m}_b}=\sqrt{\frac{1\pm\eta_b}{2}}e^{-i\varphi_b/2}\ket{0}\pm\sqrt{\frac{1\mp\eta_b}{2}}e^{i\varphi_b/2}\ket{1}$$

and 

$$\vec{m}_b=\sqrt{3}\left(\sqrt{1-\eta_b^2}\cos\varphi_b, \sqrt{1-\eta_b^2}\sin\varphi_b,\eta_b\right)$$

for

\begin{align}\nonumber
& \vec{m}_1=\left(+1,+1,+1\right), \quad \vec{m}_2=\left(+1,-1,-1\right),\\
& \vec{m}_3=\left(-1,+1,-1\right), \quad \vec{m}_4=\left(-1,-1,+1\right).
\end{align}

From Lemma \ref{lemma:state-and-measurement-trivial-for-quantum} we have that 
\begin{align*}
p^\text{Q}_{{\rm guess}}(X|E,\mathbb{I}/2,\{M_S^{x,\theta}\}_x)=1
\end{align*}
for $\{M_S^{x,\theta}\}_x:=\{\Tr_A[\ket{\phi^\theta_x}\bra{\phi^\theta_x}_{SA}]/2\}_x$ and for all $\theta\in[0,\pi/2]$.

\medskip

We implemented the SDP in Lemma \ref{lemma:pm-simulability} and obtained that the POVMs $\{M_S^{x,\theta}\}_x$ are not a convex combination of PMs for values of $\theta\in[0,\pi/10]$ \cite{pythoncode}. Although we conjecture that this holds for every $\theta\in[0,\pi/10]$, given that numerically one can only check finitely many, for this proof we set $\theta=0$ (one of those values that we have checked).  Therefore, by Lemma \ref{lemma:classical-pguess-max-with-projectors}, we have that $$p^\text{C}_{{\rm guess}}(X|\Lambda,\mathbb{I}/2,\{M_S^{x}\}_x)<1$$ for $\{M_S^x\}=\{M_S^{x,0}\}_x$, thus concluding the proof.
\end{proof}

\end{widetext}

\end{document}